\documentclass[journal,comsoc]{IEEEtran}
\usepackage[T1]{fontenc}
\usepackage{graphicx}

\usepackage{algorithmic}
\usepackage{algorithm}
\usepackage{amsmath}
\usepackage{amssymb}
\usepackage{textcomp}
\usepackage{amsthm}
\usepackage{wrapfig}
\usepackage{lipsum}
\usepackage{xcolor}
\graphicspath{ {images/} }
\usepackage{array}
\ifCLASSOPTIONcompsoc
  \usepackage[caption=false,font=normalsize,labelfont=sf,textfont=sf]{subfig}
\else
  \usepackage[caption=false,font=footnotesize]{subfig}
\fi
\usepackage{stfloats}
\usepackage{url}
\usepackage{balance}
\newtheorem{theorem}{Theorem}

\usepackage[cmintegrals]{newtxmath}

\usepackage{cite}
\usepackage{amsmath,amssymb,amsfonts}
\usepackage{algorithmic}
\usepackage{graphicx}
\usepackage{textcomp}
\usepackage{xcolor}
\begin{document}

\title{Extreme Value Theory Based Rate Selection for Ultra-Reliable Communications\\
}


\author{Niloofar~Mehrnia,~\IEEEmembership{Student Member,~IEEE,}
        Sinem~Coleri,~\IEEEmembership{Fellow,~IEEE}
\thanks{Copyright (c) 2015 IEEE. Personal use of this material is permitted. However, permission to use this material for any other purposes must be obtained from the IEEE by sending a request to pubs-permissions@ieee.org.}
\thanks{Niloofar Mehrnia and Sinem Coleri are with the Department of Electrical and Electronics Engineering,
Koc University, Istanbul, Turkey (e-mail: nmehrnia17@ku.edu.tr; scoleri@ku.edu.tr).}
\thanks{Niloofar Mehrnia is also with Koc University Ford Otosan Automotive Technologies Laboratory (KUFOTAL), Sariyer, Istanbul, Turkey, 34450.}
\thanks{Sinem Coleri acknowledges the support of Ford Otosan.}
}

\maketitle

\begin{abstract}
Ultra-reliable low latency communication (URLLC) requires the packet error rate to be on the order of $10^{-9}$-$10^{-5}$. Determining the appropriate transmission rate to satisfy this ultra-reliability constraint requires deriving the statistics of the channel in the ultra-reliable region and then incorporating these statistics into the rate selection. In this paper, we propose a framework for determining the rate selection for ultra-reliable communications based on the extreme value theory (EVT). We first model the wireless channel at URLLC by estimating the parameters of the generalized Pareto distribution (GPD) best fitting to the tail distribution of the received powers, i.e., the power values below a certain threshold. Then, we determine the maximum transmission rate by incorporating the Pareto distribution into the rate selection function. Finally, we validate the selected rate by computing the resulting error probability. Based on the data collected within the engine compartment of Fiat Linea, we demonstrate the superior performance of the proposed methodology in determining the maximum transmission rate compared to the traditional extrapolation-based approaches. 
\end{abstract}


\section{Introduction}
\IEEEPARstart{U}{ltra}-reliable low latency communication (URLLC) is one of the most important features of the fifth-generation ($5$G) networks with the aim of supporting mission critical applications, such as remote control of robots, remote surgery, autonomous vehicles and vehicular teleoperation applications \cite{interface_01}-\nocite{urllc_02}\cite{urllc_05}. At URLLC, the packet error rate (PER) is guaranteed to be as low as $10^{-9}$-$10^{-5}$ to address the strict reliability constraint, while the acceptable latency is on the order of a few milliseconds \cite{interface_01}-\nocite{5G_01}\nocite{interface_02}\nocite{urllc_02}\nocite{MehrniaTWC}\cite{MehrniaTVT}. Establishing a URLLC system necessitates the statistical modeling of the wireless channel tail quantifying the statistics of extreme events \cite{urllc_05}-\nocite{MehrniaTWC}\cite{MehrniaTVT}, and the transmission strategies in the ultra-reliable region \cite{urllc_05}.

Previous studies on the statistical modeling of the wireless channel for ultra-reliable communications extrapolate a wide range of practically important channel models to the ultra-reliability region \cite{urllc_02}, \cite{urllc_05}. A simple power-law expression is proposed to estimate the tail of the cumulative distribution function (CDF) of the received power by extrapolating the commonly used practical channel models.
Only recently, we have demonstrated that these extrapolated distributions are not accurate in the ultra-reliable region and can result in several orders of magnitude difference in the estimated packet error probabilities \cite{MehrniaTWC}-\cite{MehrniaTVT}. We have proposed the usage of fading statistics from extreme value theory (EVT) and developed a novel methodology to derive these statistics efficiently with minimum amount of data. 
Nevertheless, none of these studies determine the maximum transmission rate in URLLC or assess the system reliability by considering ultra-reliable channel statistics.

Calculation of the reliability and rate selection for ultra-reliable communications has been addressed in the context of proposing new channel parameters and new performance measures. The definition of coherence time/distance has been modified as the time or distance over which a channel is predictable with a given reliability in \cite{urllc_07}-\cite{urllc_08}. New performance measures are defined as average reliability for dynamic environments where the channel changes frequently; and probably correct reliability for the static environments where the channel statistics remain constant for a large enough time interval \cite{urllc_05}.
However, these reliability measures have been derived by using the extrapolation of the traditional average statistic based channel models, which may not be accurate in the ultra-reliable region. 
Deriving outage probability for assessing the reliability of a wireless channel model is immensely important as uncertainty may degrade the communication performance by several orders of magnitude \cite{urllc_05}. It is worth noting that the asymptotic idea of outage probability, which is widely employed in wireless communication systems, is an accurate performance criterion even in the finite blocklength context \cite{blocklength}.

The goal of this study is to propose a novel EVT-based framework for the estimation and validation of the optimal transmission rate for ultra-reliable communications. 
The original contributions of the paper are listed as follows:

\begin{itemize}
    \item We propose a novel framework based on the modeling of the tail distribution of the channel by using the GPD, determination of the optimum transmission rate by using the estimated values of the GPD parameters, and then assessment of the system reliability by means of the outage probability metric. 
    \item We formulate the rate selection function of GPD by incorporating the estimated Pareto parameters into the optimum transmission rate that guarantees the ultra-reliability constraints. 
    \item We assess the reliability of the system by calculating the outage probability based on the rate selection function for the GPD.
    \item We demonstrate the superiority of the proposed methodology in terms of the reliability assessment, compared to the conventional method based on the extrapolation of the average statistics to the ultra-reliable region, over the data collected within the engine compartment of Fiat Linea under various engine vibrations and driving scenarios.
\end{itemize}

The rest of the paper is organized as follows. Section \ref{sec:system_model} describes the system model and assumptions considered throughout the paper. Section~\ref{sec:framework} presents the EVT-based framework for determining and validating the rate selection for the tail distribution of the channel by using GPD.
Section \ref{sec:numerical_results} provides the channel measurement setup and the performance evaluation in determining the optimum rate and outage probability. Finally, concluding remarks and future works are given in Section \ref{sec:conclusions}.
\section{System Model}
\label{sec:system_model}
In this study, we consider an ultra-reliable communication system experiencing significant fading, resulting in extremely low received power values. We assume that the transmitter (Tx) sends a packet to a receiver (Rx) at the rate $R$ over an unknown channel. The receiver estimates the parameters of the GPD fitted to the channel tail distribution by applying EVT to the received power values. Then, the transmitter determines the transmission rate based on these estimated parameter values.

We assume that the channel is stationary, i.e., the distribution function of the received powers does not change over time, and the parameters specifying the distribution class are fixed over time. 
In the case of non-stationarity based on the results of the Augmented Dickey-Fuller (ADF) test, the external factors varying the parameters of the GPD are determined such that the sequence is divided into $M$ groups, in which the channel can be considered stationary, as detailed in \cite{MehrniaTVT}.
Also, the transmit power is fixed and known in advance. In such a case, using the received signal power is equivalent to using the squared amplitude of the channel state information \cite{urllc_05}, \cite{MehrniaTWC}.

To study the outage probability and define the maximum transmission rate, we assume that the dominant source of error in block fading channels is the link outage, where we neglect other sources such as environmental noise. This model has been shown to be very suitable in the transmission of short packets in URLLC scenarios \cite{urllc_02}-\cite{urllc_05}, \cite{urllc_noise}. Prior to transmission, at the training phase and after collecting the stationary channel samples by Rx, the channel sequence is converted into $n$ independent and identically distributed (i.i.d.) sample sequence by using declustering method \cite{MehrniaTWC}, \cite{evt_04}, denoted by $X^n = \{x_1, ...,x_n\}$, where $x_{i}$ is the $i^{th}$ i.i.d. sample, $i \in \{1,...,n\}$. Let $F$ be the CDF of the training samples $X^n$. The training phase uses either the dedicated pilot signals and training sequences or, the history of previous data transmissions and the associated feedback. 
The link outage is defined as
\begin{equation}
\label{eqn:outage}
    R(X^{n})>\log_2(1+Z),
\end{equation}
where $R(X^{n})$ denotes the optimum transmission rate estimated based on the $n$ training received power samples, and $Z$ is the received power from the test samples. Note that in (\ref{eqn:outage}), unit bandwidth is assumed; hence, $R(.)$ represents the transmission rate per unit bandwidth, i.e., spectral efficiency, in bits/sec. Then, according to (\ref{eqn:outage}), the outage probability at transmission rate $R(X^{n})$ is defined as
\begin{equation}
\label{eqn:outageprob}
    p_{F}(R(X^{n})) = P\big[ R(X^{n})>\log_2(1+Z) \big]. 
\end{equation}

The goal of ultra-reliable communication is to choose the maximal rate that meets a predetermined reliability constraint, such that
\begin{equation}
\label{eqn:pfRxn}
    p_F(R(X^{n})) \leq \epsilon.
\end{equation}
Guaranteeing constraint (\ref{eqn:pfRxn}), the transmitter determines the maximum rate as a function of $F$ as follows \cite{urllc_05}:
\begin{align}
\label{eqn:eoutage}
\begin{split}
   R_{\epsilon}(F) = sup\big\{ R(X^{n}) \ge 0 : p_{F}(R(X^{n})) \le \epsilon\big\}\\
   \hfill \hspace{0.3cm}
   = \log_{2} \big(1+F^{-1}(\epsilon)\big),
\end{split}
\end{align}
where $R_{\epsilon}(F)$ denotes $\epsilon$-outage capacity, and $F^{-1}(\epsilon)$ is the $\epsilon$-quantile of $F$.
\section{Rate Selection Framework}
\label{sec:framework}
In our EVT-based framework, the rate selection function of GPD is formulated by incorporating the estimated Pareto parameters, i.e., scale and shape parameters, into the derivation of the optimum transmission rate guaranteeing ultra-reliability. First, received power samples are collected in the training phase.
Afterwards, the training samples are converted to $n$ i.i.d. samples by using the declustering method \cite{MehrniaTWC}, \cite{evt_04}, as an input to the EVT process. Then, GPD is fitted to the lower tail of the i.i.d. sequence.
Modeling the lower tail is followed by formulating the rate selection function of GPD. Finally, the channel reliability is assessed by comparing the estimated outage probability obtained in the test phase with the targeted PER. The proposed algorithm is depicted in Fig.~\ref{fig:framework} and explained in detail next.
\begin{figure*}[ht] 
\centering{
\includegraphics[scale=0.7]{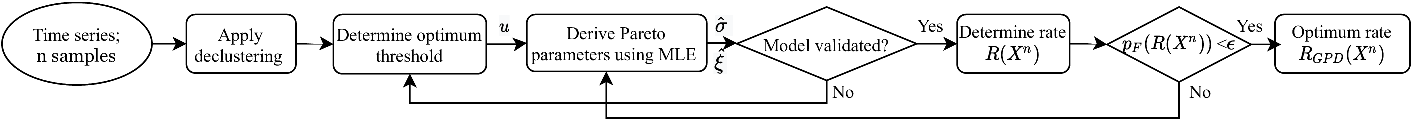}
\caption{Flowchart of the proposed rate selection framework.}
\label{fig:framework}}
\end{figure*}
\subsection{Channel Estimation}
\label{sec:channelestimation}
EVT provides a robust framework for analyzing the statistics of extreme events happening rarely through modeling the probabilistic distribution of the values exceeding a given threshold by using the GPD \cite{MehrniaTWC}, \cite{evt_04}. Assume that $X^n=\{x_1,...,x_n\}$ is an i.i.d. stationary sequence of received powers $x_{i}$ for $i \in \{1,...,n\}$. Accordingly, the probabilistic distribution of the power values exceeding a given threshold $u$ can be expressed as
\begin{equation}
\label{eqn:gpddist}
    F_{u}(y) = 1-\big[1+\frac{\xi y}{\tilde{\sigma}_{u}}\big]^{-1/\xi},
\end{equation}
where $y$ is a non-negative value denoting the exceedance below threshold $u$, i.e., ($y=u-X$, $X$ is any $x_i$ below threshold $u$); and $\xi$ and $\tilde{\sigma}_{u}=\sigma+\xi(u-\mu)$ are shape and scale parameters of the GPD, respectively, and $\mu$ and $\sigma$ are the location and scale parameters of the generalized extreme value (GEV) distribution fitted to the CDF of $M_n = max \{ x_1,...,x_n\}$, respectively \cite[Theorem~1]{MehrniaTWC}, \cite{evt_11}.

To model the tail distribution of the received power sequence, the measured samples are converted into an i.i.d. sequence by removing their dependency using the declustering approach \cite{MehrniaTWC}. Then, EVT is applied to the i.i.d. samples for optimum threshold determination and the estimation of Pareto distribution parameters by using the MLE. The optimum threshold is determined by utilizing two complementary methods, mean residual life (MRL) and parameter stability methods \cite{MehrniaTWC}, \cite{evt_04}. The MLE estimates of the Pareto distribution parameters are formulated as follows.

\begin{theorem}
Let GPD($\sigma$,$\xi$) be the Pareto model fitted to the training samples $X^{n} = \{x_1,x_2,...,x_n\}$. Then, the MLE of $\sigma$ and $\xi$ can be obtained as
\begin{equation}
\label{eqn:estimatedparam}
    \begin{array}{ll}
     \hat{\xi} = \frac{1}{k}\big[\sum_{i=1}^{k}\log\big(1+\hat{\theta}y_i\big)\big];  \\
     \hat{\sigma} = \frac{\hat{\xi}}{\hat{\theta}}
\end{array}
\end{equation}
where $k$ is the number of samples in the tail, i.e., exceeding the optimum threshold; $y_i = u-x_{i}$ for all $x_{i}<u$, where $i\in \{1,...,n\}$; and $\hat{\theta}$ is the root of the following equation:
\begin{equation}
\label{eqn:psi}
    \Psi_{k}(\hat{\theta}) = \Big[\frac{1}{k} \sum_{i=1}^{k} (1+\hat{\theta}y_i)^{-1}  (\frac{1}{k} \sum_{i=1}^{k} \log(1+\hat{\theta}y_i) +1 ) \Big]{-1}.
\end{equation}
\end{theorem}

\begin{proof}
The log-likelihood function corresponding to the Pareto distribution fitted to the tail of i.i.d. sequence $X^{n} = \{x_1,x_2,...,x_n\}$ is defined as
    $l_{Y}(\hat{\sigma},\hat{\xi})=\sum_{i=1}^{k}l_{y_{i}}(\hat{\sigma},\hat{\xi})$,
where $Y=\{y_1,y_2,...,y_k\}$; $\hat{\sigma}$ and $\hat{\xi}$ are the estimated scale and shape parameters, respectively; and $l_{y_i}(\hat{\sigma},\hat{\xi})$ is expressed as \cite{evt_04}
\begin{equation}
    \label{eqn:loglokelihoodGPD}
    l_{y_i}(\hat{\sigma},\hat{\xi}) = \log (\frac{1}{\hat{\sigma}}) - (1+ \frac{1}{\hat{\xi}}) \log (1+\frac{\hat{\xi}}{\hat{\sigma}}y_i).
\end{equation}
By taking the partial derivative of (\ref{eqn:loglokelihoodGPD}) with respect to $\hat{\sigma}$ and $\hat{\xi}$, and making them equal to $0$, we obtain the estimation of scale and shape parameters of GPD, respectively. 
Equating the partial derivatives to $0$, we come up with the following equations:

\begin{subequations}
\begin{equation}
\label{eqn:appendixxi}
    \hat{\xi} = \frac{1}{k} \sum_{i=1}^{k} \log (1+\frac{\hat{\xi}}{\hat{\sigma}}y_i),
\end{equation}
\begin{equation}
\label{eqn:appendixxirev}
    \frac{1}{1+\hat{\xi}} = \frac{1}{k} \sum_{i=1}^{k} \frac{1}{1+\frac{\hat{\xi}}{\hat{\sigma}}y_i}.
\end{equation}
\end{subequations}
Substituting (\ref{eqn:appendixxi}) into (\ref{eqn:appendixxirev}), we end up with a new equality as $\Psi_{k}(\hat{\theta})$ function expressed in (\ref{eqn:psi}).
Then, the MLE of $\theta$, i.e., $\hat{\theta}$, is the root of (\ref{eqn:psi}) and accordingly, the estimations of $\xi$ and $\sigma$ are obtained as shown in (\ref{eqn:estimatedparam}).
\end{proof}

The validity of the estimated Pareto model is then assessed by using probability plots including probability/probability (PP) plot and quantile/quantile (QQ) plot. We refer the readers to \cite{MehrniaTWC} and \cite{MehrniaTVT} for more details on the channel modeling methodology of the extreme events at URLLC.

\subsection{Rate Selection}
\label{sec:rateselectionEVT}
Rate selection requires the estimation of the GPD parameters and calculation of the rate as a function of these parameters.
Therefore, in order to determine the $\epsilon$-outage in (\ref{eqn:eoutage}), we estimate $F^{-1}(\varepsilon_{n})$ as
 \begin{equation}
     \label{eqn:gpdinv}
     {F}_{u}^{-1}(\varepsilon_{n}) = u + \frac{\sigma_{u}}{\xi}\big[1 - \varepsilon_{n}^{-{\xi}} \big],
 \end{equation}
where ${F}_{u}^{-1}(.)$ denotes the inverse of $F_{u}(.)$; $u$ is the optimum threshold for GPD; and $\varepsilon_{n}$ quantile is the probability whose associated quantile is of interest \cite{MehrniaTWC}. The shape and scale parameters of GPD in (\ref{eqn:gpdinv}) are the MLE of the GPD parameters fitted to the tail of the training samples $X^{n} = \{x_1,x_2,...,x_n\}$.

\begin{theorem}
Let $\hat{\sigma}$ and $\hat{\xi}$ be the MLE of the scale and shape parameters of the GPD($\sigma$,$\xi$) fitted to the tail distribution of the training samples $X^{n} = \{x_1,x_2,...,x_n\}$, respectively. 
Then, the maximum transmission rate is defined as
\begin{equation}
    \label{eqn:maxrate}
    R_{GPD}(X^n) = \log_{2} \big(1+u + \frac{\hat{\sigma}}{\hat{\xi}}\big[1 - \varepsilon_{n}^{-{\hat{\xi}}} \big]\big),
\end{equation}
where $u$ is the optimum threshold for the channel tail estimation; and $\varepsilon_n$ is the outage probability.
\end{theorem}

\begin{proof}
The transmission rate for a specific training sample $X^{n} = \{x_1,x_2,...,x_n\}$ is defined as \begin{equation}
\label{eqn:Rxngeneral}
    R(X^n) = \log_{2}\big(1+\hat{F}^{-1}(\varepsilon_n)\big),
\end{equation}
where the $\hat{F}^{-1}(\varepsilon_n)$ is the estimate of $\varepsilon_n$-quantile of $F$, for any distribution $F$. The goal is to find $\varepsilon_n$ such that $R(X^n)$ is maximized and (\ref{eqn:pfRxn}) is satisfied. Substituting (\ref{eqn:gpdinv}) in (\ref{eqn:Rxngeneral}), the maximum transmission rate guaranteeing a certain reliability, i.e., error probability $\epsilon$, is defined as (\ref{eqn:maxrate}).
\end{proof}
\subsection{Validation of Selected Rate}
\label{sec:outageprobEVT}
The average outage probability, i.e., the average of error probability due to the outage of the GPD fitted to the extreme values exceeding a given threshold, is obtained by substituting (\ref{eqn:maxrate}) in outage probability equation (\ref{eqn:outageprob}), and taking the expectation as follows:
 \begin{align}
    \begin{split}
    \label{eqn:reliabilitymeasure}
    p_{F}(R_{GPD}(X^{n})) =\\
    E\Big[ P\big[ \log_{2} (1+u+\frac{\hat{\sigma}}{\hat{\xi}}(1-\varepsilon_{n}^{-\hat{\xi}}) > \log_{2}(1+Z) \big] \Big]=\\
    E\Big[ P\big(-\frac{\hat{\sigma}}{\hat{\xi}}(1-\varepsilon_{n}^{-\hat{\xi}}) < u-Z \big) \Big]=\\
    E \Big[\big(1-\frac{{\xi} }{{\sigma}} (\frac{\hat{\sigma}}{\hat{\xi}}(1-\varepsilon_{n}^{-\hat{\xi}})) \big)^{-1/{\xi}} \Big],
\end{split}
\end{align}
where $E[.]$ denotes the expectation function; $\hat{\xi}$ and $\hat{\sigma}$ are the MLE of estimated shape and scale parameters fitted to the training samples, respectively; and $\xi$ and $\sigma$ are the shape and scale parameters of the GPD, given that $F_u(.)$ is perfectly known and GPD is fitted to the sample sequence including the data in training and test phases. By referring to the reliability constraint expressed in (\ref{eqn:pfRxn}), the outage probability obtained from (\ref{eqn:reliabilitymeasure}) should be less than or equal to the $\epsilon$. Therefore, the maximum allowed error probability, $\varepsilon_n$, can be calculated with respect to the targeted PER $\epsilon$, as follows:
\begin{equation}
    \varepsilon_n = \big[ 1-\frac{\hat{\xi}}{\hat{\sigma}} \big(\frac{\sigma}{\xi}\big(1-\epsilon^{-\xi}\big)\big)\big]^{-\frac{1}{\hat{\xi}}}.
\end{equation}
If $R_{GPD}(X^n)$ is a valid transmission rate, the corresponding outage probability $p_{F}(R_{GPD}(X^{n}))$ is expected not to exceed the targeted reliability $\epsilon$, according to constraint (\ref{eqn:pfRxn}).

The complexity of the proposed rate selection framework is $O(n)$, similar to the traditional extrapolation approach \cite{urllc_05}, \cite{urllc_07}-\cite{urllc_08}, since the optimum transmission rate is determined based on $n$ training samples.

\section{Numerical Results}
\label{sec:numerical_results}
The goal of this section is to evaluate the performance of the proposed methodology in determining the maximum transmission rate for URLLC, evaluating the system reliability by using the outage probability metric, comparing it to the traditional extrapolation-based methods for the Rayleigh fading in the estimation of the transmission rate, and illustrating the impact of channel mismatch on the reliability performance of the system. It is worth mentioning that throughout this section, the Rayleigh distribution refers to the channel model under the Rayleigh fading.

To obtain the extrapolated Rayleigh curve, we fit the Rayleigh distribution to the first $10^3$ samples from the received power sequence and upon estimating the Rayleigh parameter, we compute the transmission rate as expressed in \cite[Equation~23]{urllc_05}. Then, we extrapolate the fitted Rayleigh distribution toward the ultra reliable region and compute the corresponding outage probability as expressed in \cite[Equation~26]{urllc_05}, yet assuming that the transmission rate is the Rayleigh rate computed based on \cite[Equation~23]{urllc_05}.

To represent the impact of channel mismatch on the system reliability, we fit the Rayleigh distribution to the received power values and then, compute the maximum rate as expressed in \cite[Equation~23]{urllc_05}. Then, we assume that the channel is no longer Rayleigh, i.e., $F$ in constraint (\ref{eqn:outageprob}) is different than Rayleigh (the true distribution $F$ is estimated by the GPD at different thresholds as it refers to the distribution of the channel tail); yet Tx assumes that the channel is Rayleigh and hence, setting the rate as in \cite[Equation~23]{urllc_05} with $\varepsilon_n$ obtained as expressed in \cite[Equation~27]{urllc_05}. Due to mismatch, the estimated transmission rate w.r.t. the Rayleigh distribution does not converge to $R(X^n)$ computed w.r.t. the true distribution $F$. Additionally, the outage probability is computed such that the channel tail distribution is estimated by the GPD while the system operating at the Rayleigh transmission rate obtained by \cite[Equation~23]{urllc_05}.

The measured channel data were collected within the engine compartment of Fiat Linea at $60$ GHz. The locations of the transmitter and receiver antennas are selected out of the possible locations for the wireless sensors located within the engine compartment, namely locations $5$ and $13$ in \cite[Fig.~$1$]{vehicular_02}, such that the effect of the engine vibration is observed in the received power, as shown in Fig.~\ref{fig:setup}. A Vector Network Analyzer (VNA) (R$\And$S$\textsuperscript{\textregistered}$ ZVA$67$) is connected to the transmitter and receiver via the R$\And$S$\textsuperscript{\textregistered}$ ZV-Z$196$ port cables with maximum $4.8$ dB transmission loss. The horn transmitter and receiver antennas with a nominal $24$ dBi gain and $12^\circ$ vertical beam-width operate at $50$-$75$ GHz. We have captured about $10^{6}$ successive samples for $30$ minutes with a time resolution of $2$ ms. We use \textsc{MATLAB} for the implementation of the proposed algorithm. Due to the existence of a non-stationarity trend among the samples, the measured samples are categorized into two stationary groups according to the engine vibration. Group one includes bunches of $10^3$ successive samples all above $-12$ dBm. If any sample within $10^3$ samples is less than $-12$ dBm, the set of successive samples is assigned to group two.
\begin{figure}[h] 
\centering{
\includegraphics[width=0.6\columnwidth, height=3cm]{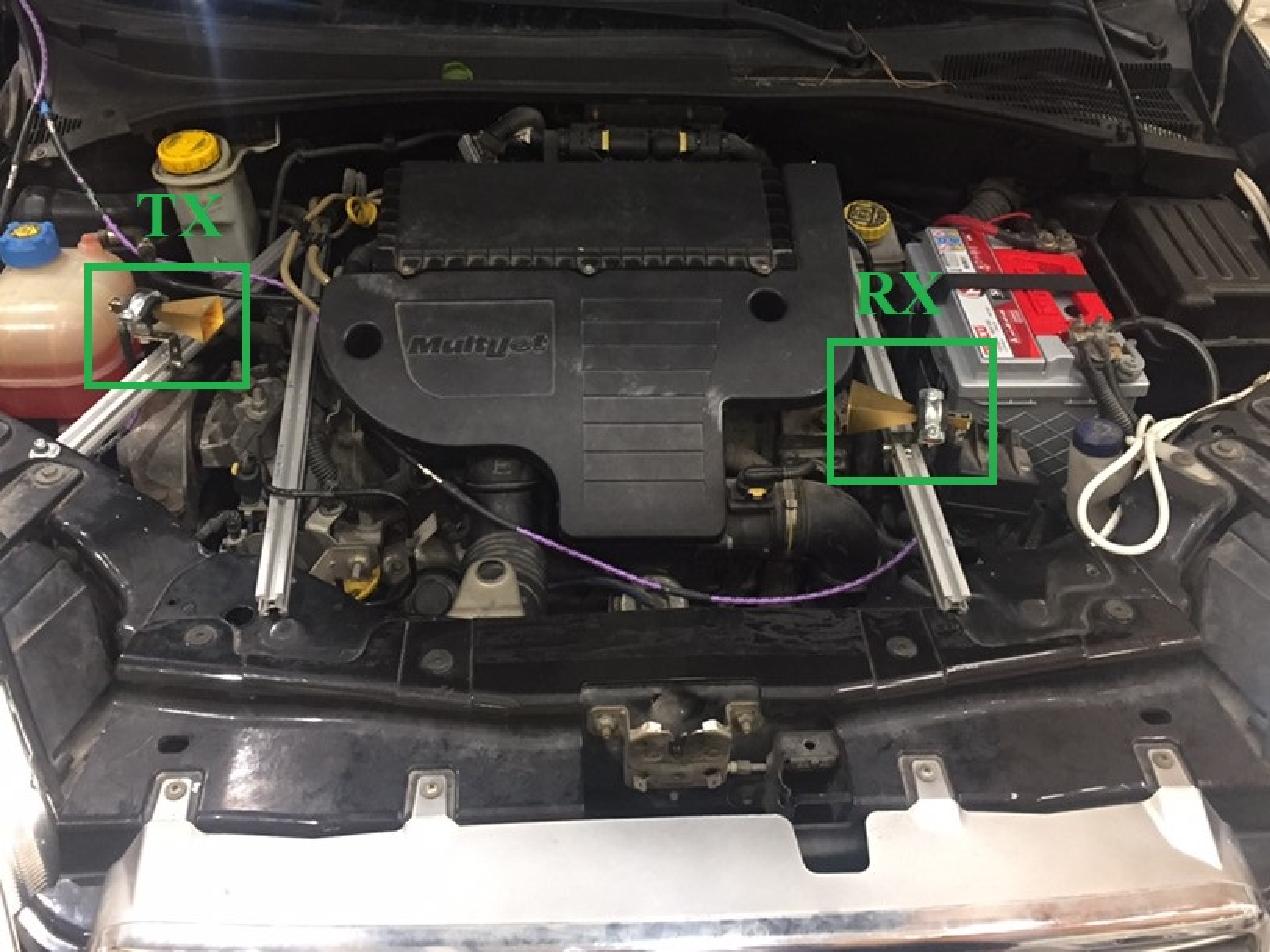}
\caption{Measurement setup with the transmitter (TX) and receiver (RX) antennas located in the engine compartment of Fiat Linea.}
\label{fig:setup}}
\end{figure}

\subsection{Optimum transmission rate}
To determine the maximum transmission rate at URLLC, we plot the normalized $R_{GPD}(X^n)$, denoted by $w$, for different sample numbers. The normalization is performed with respect to the optimal throughput, given that $F_u(.)$ is perfectly known, i.e., the last $R_{GPD}(X^n)$ estimated based on the whole set of data samples, not only the training ones. 

Fig.~\ref{fig:wg1g2} shows the normalized rate $w$ for GPD fitted to the filtered i.i.d. samples of groups $1$ and $2$ at different sample numbers, threshold values, and error rates $\epsilon \in \{10^{-3},10^{-4},10^{-5}\}$. 
Fig.~\ref{fig:wg1_10dbm} shows that the normalized rate selection function $w$ converges faster to $1$, meaning that the transmission rate converges to its optimum value, at lower targeted PER $\epsilon = 10^{-3}$. 
Comparing Fig.~\ref{fig:wg1_10dbm} and Fig.~\ref{fig:wg1_5dbm} at each targeted PER $\epsilon$, the normalized rate in Fig.~\ref{fig:wg1_5dbm} approaches to $1$ significantly faster than that shown in Fig.~\ref{fig:wg1_10dbm}. 
This is expected since by relaxing the threshold from $-10$ dBm to $-5$ dBm, determining GPD and its corresponding transmission rate requires the collection of less samples in the training phase. Additionally, by increasing the threshold, at lower sample numbers, greater portion of received power with higher range and variety is included in the tail and therefore, the GPD parameters and the transmission rate estimated at lower sample number remain valid for higher sample numbers. Fig.~\ref{fig:wg2_25dbm} also demonstrates that the convergence of $w$ function to $1$ is slow for targeted PER $\epsilon = 10^{-5}$.
The fluctuations in Fig.~\ref{fig:wg2_25dbm} is due to the low number of collected samples for estimating the GPD and its corresponding rate function.

Figs.~\ref{fig:wg1_10dbm}, \ref{fig:wg1_5dbm}, and \ref{fig:wg2_25dbm} also illustrate the normalized rate $w$ for GPD compared to the extrapolation based method for groups $1$ and $2$ of data and for different thresholds of $-5$~dBm and $-10$~dBm at group $1$, and $-25$~dBm at group $2$. We observe that the normalized rate under the GPD assumption converges faster or simultaneously to the optimum rate, compared to that of the extrapolated approach for both groups (referring to Figs.~\ref{fig:wg1_10dbm}, \ref{fig:wg1_5dbm}, and \ref{fig:wg2_25dbm}), and under different thresholds (referring to Figs.~\ref{fig:wg1_10dbm} and \ref{fig:wg1_5dbm}). The normalized rate of extrapolated approach in Fig.~\ref{fig:wg1g2} corresponds to all $\epsilon$ values, as the normalized rate under the extrapolated Rayleigh rate-selection function appears to be (almost) independent from the targeted error probability $\epsilon$, implying that the rate-selection function obtained for a given $\epsilon$ is valid for the others as well, as also stated in \cite{urllc_05}.

Figs.~\ref{fig:wg1_10dbm}, \ref{fig:wg1_5dbm}, and \ref{fig:wg2_25dbm} besides show the impact of mismatch between the estimated channel and the true one. The proposed method in Fig.~\ref{fig:wg1g2} performs significantly better than the case where the true channel distribution is actually modeled by the GPD while Rayleigh distribution was assumed.
As also stated in \cite{urllc_05}, the convergence of the rate is heavily affected by the model mismatch, degrading the normalized rate to almost $0$ for all $\epsilon$ values. However, the effect of model mismatch can be reduced by increasing the sample number to the higher values.
\begin{figure*}[ht]
\centering
\captionsetup[subfigure]{labelformat=empty}
     \begin{center}
        \subfloat[(a)]{%
            \label{fig:wg1_10dbm}
            \includegraphics[width=0.6\columnwidth]{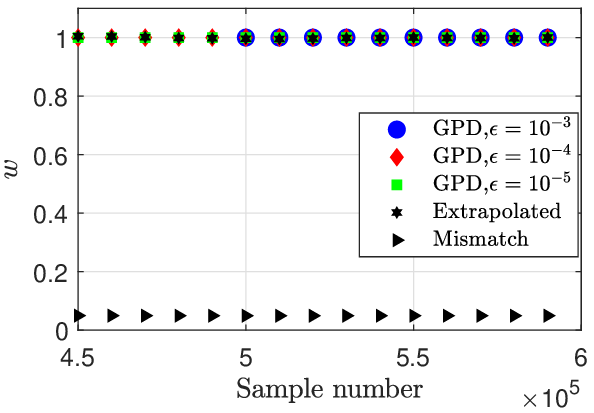}
        }
        \subfloat[(b)]{%
            \label{fig:wg1_5dbm}
            \includegraphics[width=0.6\columnwidth]{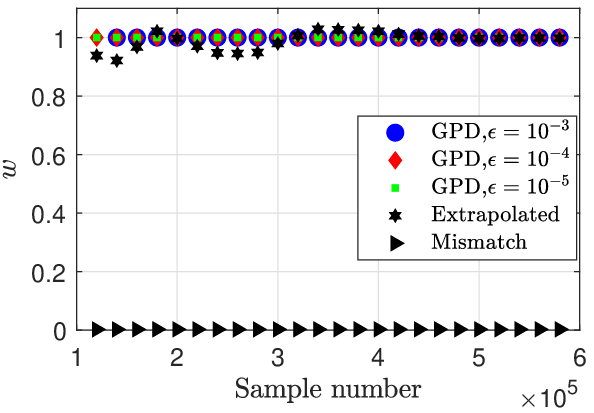}
        }
        \subfloat[(c)]{%
            \label{fig:wg2_25dbm}
            \includegraphics[width=.6\columnwidth]{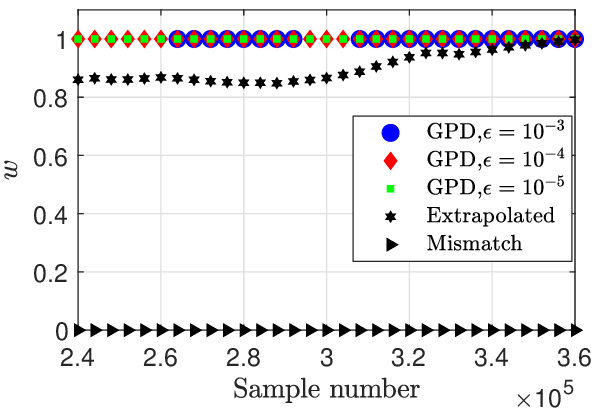}
        }\\
    \end{center}
    \caption{Normalized transmission rate for: (a) group $1$ at $u=-10$ dBm, (b) group $1$ at $u=-5$ dBm, and (c) group $2$ at $u=-25$ dBm. The single \textit{Extrapolated} and \textit{Mismatch} plots correspond to all $\epsilon$ values. \textit{Extrapolated} and \textit{Mismatch} plots are both based on the Rayleigh assumptions.}
   \label{fig:wg1g2}
\end{figure*}

\subsection{Reliability assessment}
Fig.~\ref{fig:comparison} illustrates the reliability performance of the fitted GPD based on the calculated outage probability for a range of thresholds for $2$ groups of data. In this figure, the outage probability in (\ref{eqn:reliabilitymeasure}) is plotted for fixed error probability $\epsilon \in \{10^{-3},10^{-4},10^{-5}\}$ for large enough $n$ to ensure convergence \cite{urllc_05}. This demonstrates that at both groups and for all error probabilities, the reliability constraints in (\ref{eqn:pfRxn}) are never violated. 
\begin{figure}[h]
\centering
\captionsetup[subfigure]{labelformat=empty}
     \begin{center}
        \subfloat[(a)]{%
            \label{fig:comp_rg1}
            \includegraphics[width=0.47\columnwidth]{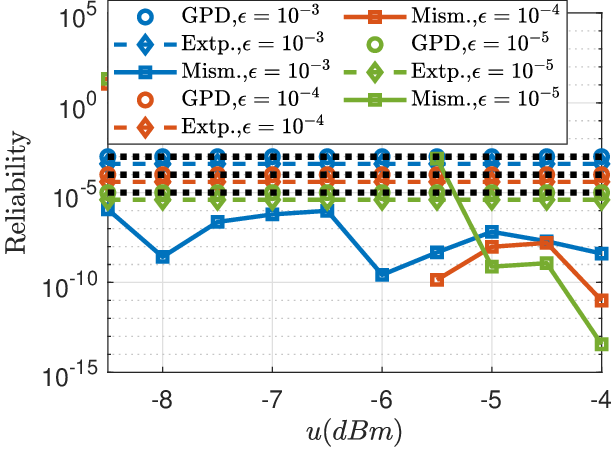}
        }
        \subfloat[(b)]{%
            \label{fig:comp_rg2}
            \includegraphics[width=0.47\columnwidth]{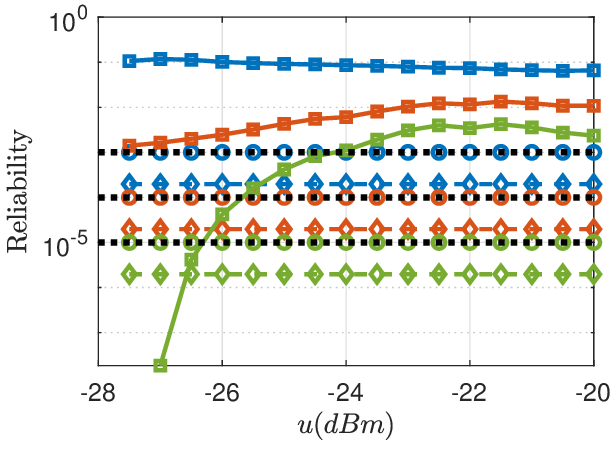}
        }\\ 
        
    \end{center}
    \caption{Reliability measure of the proposed model and the traditional extrapolation approach for: (a) group $1$, and (b) group $2$; The dashed-black lines are the reference lines at targeted error rate $\epsilon \in \{10^{-3},10^{-4},10^{-5}\} $. \textit{Extrapolated} and \textit{Mismatch} plots are both based on the Rayleigh assumptions.}
   \label{fig:comparison}
\end{figure}

Figs.~\ref{fig:comp_rg1} and \ref{fig:comp_rg2} also show the reliability performance of the proposed Pareto model compared to the conventional extrapolation-based method by comparing their outage probabilities for $\epsilon \in \{10^{-3},10^{-4},10^{-5}\}$. Since in the traditional approach, the outage probability is independent of any GPD thresholds, we plot a single outage probability obtained by the extrapolated method at all the threshold values. In Fig.~\ref{fig:comp_rg1}, the reliability of the extrapolated approach at each $\epsilon$ value is almost identical to that of the proposed approach and equal to the desired error rate. However, in Fig.~\ref{fig:comp_rg2}, the achieved reliability obtained by the extrapolated approach differs from the proposed GPD-based approach, as well as the targeted $\epsilon$. It is due to the fact that in Group $2$, the system suffers from more extreme values with lower received power. In such a situation, the channel tail needs to be estimated by the GPD and not the extrapolation. 

Figs.~\ref{fig:comp_rg1} and \ref{fig:comp_rg2} illustrate the effect of mismatch on the reliability of the system through comparing their outage probabilities for $\epsilon \in \{10^{-3},10^{-4},10^{-5}\}$. The proposed method in Fig.~\ref{fig:comp_rg1} has been demonstrated to perform significantly better than the case where the true channel distribution differs from the estimated Rayleigh distribution. The effect of mismatch is more clear for the low $\epsilon$ values and at higher thresholds of the GPD. However, the proposed method in Fig.~\ref{fig:comp_rg2} performs much better for all $\epsilon$ values and thresholds. On the other hand, while with the proposed GPD approach, the reliability is always equivalent to the required error rate $\epsilon$, the reliability of the mismatch case with Rayleigh rate assumption deviates from the desired $\epsilon$ almost at all thresholds.

\section{Conclusions}
\label{sec:conclusions}
In this paper, we propose an EVT-based rate selection framework for ultra-reliable communications. First, we represent the channel by GPD distribution and estimate its parameters. Then, we determine the maximum transmission rate of the estimated channel and validate the selected rate by assessing the resulting error probability. The achieved reliability from the proposed EVT-based approach outperforms the traditional methods based on the extrapolation of average statistic channel models. Additionally, the GPD threshold has significant impact on the minimum required number of samples to achieve a certain reliability order. In the future, we plan to extend the proposed framework for the EVT analysis of the non-stationary processes by implementing a real time algorithm that refines the estimated channel parameters through the regular transmission.



\balance 

\ifCLASSOPTIONcaptionsoff
  \newpage
\fi

\bibliographystyle{ieeetr}
\bibliography{RateSelection.bib}

\end{document}